\documentclass[12pt]{article}
\usepackage[utf8]{luainputenc}
\usepackage{geometry}
\usepackage{array}
\usepackage{rotating}
\usepackage{float}
\usepackage{multirow}
\usepackage{amsmath}
\usepackage{amssymb}
\usepackage{graphicx}
\usepackage{setspace}
\usepackage{esint}
\usepackage{chngcntr}
\usepackage{enumerate}
\usepackage{graphicx}
\usepackage[authoryear]{natbib}
\usepackage{rotating}
\usepackage{changes}
\usepackage{subfig}

\usepackage{amsfonts}\usepackage{amsthm}\usepackage{fullpage}
\usepackage{multirow}\usepackage{enumerate}\restylefloat{table}\usepackage{geometry}
\counterwithin{figure}{section}
\theoremstyle{plain}
\theoremstyle{definition}
\newtheorem{proposition}{Proposition}

\newtheorem{defn}{Definition} %

\title{{Industrial symbiosis: How to apply successfully\footnote{We wish to thank Adi Ayal and Arye Hillman for comments and discussion and to Miriam Berliner for research assistance. This research was supported by the Heth Foundation for Regulation Research. The authors report there are no competing interests to declare.}}}

\author{{Limor Hatsor} \thanks{ Jerusalem College of Technology, Israel, limor.hatsor@gmail.com} \and {Artyom Jelnov} \thanks{Ariel University, Israel, artyomjel@gmail.com}}

\begin{document}
\maketitle
\begin{abstract}
 The premise of industrial symbiosis (IS) is that reusing byproducts as inputs in production is valuable for the environment. In a simple model, considering the technology parameters, we challenge this premise. Ceteris paribus, IS is environmental-friendly. However, in the equilibrium IS may be harmful for the environment. The reason is that the interest of producers to make profits by boosting production (and hence pollution), and then selling their byproducts in the market may not coincide with environmental endeavors. Using comparative static, we pinpoint a key technology parameter - the share of reused byproducts - that may hold benefits for firms, consumers, and the environment. 

\end{abstract}
\newpage
\section{Introduction}
The global population has increased significantly in recent decades along with the standard of living, resulting in growing production that involves economic costs and negative externalities for society and the environment, such as waste of natural resources, pollution, the destruction of animals' habitat, disease, and the climate change.  
In the presence of environmental concerns, circular economy is gaining  worldwide support as a promising way to reduce environmental damage, to deal with dwindling global resources and to stabilize the climate.\footnote{Many countries are adopting circular economy policies and  regulation. For example, the European Commission Action Plan allocated a budget of more than 10 billion euros for the implementation of circular economy in 2016-2020, including 1.4 billion euros for expenditures on sustainable circular industrial processes and management of resources and waste, 1.8 billion euros for the encouragement of innovative technologies and 5.3 billion euros for the passage of EU waste legislation (see EC, 2015 and  EC, 2020). France adopted an ambitious road-map for transition to a circular economy that includes a reduction of 30\% in the consumption of natural resources as a percentage of GDP per capita between 2010-2030, a reduction of 50\% in the amount of non-hazardous waste, and full recycling of plastic by 2025.}  
Circular economy calls for reducing the need to produce new raw materials by extending the life cycle of resources (byproducts, waste and energy) and products through sustainable consumption, recycling or reuse (\citealp{murray2017circular}, \citealp{wysokinska2016new}). \citet{stahel2016circular} estimates that circular economy could reduce a country’s greenhouse emissions by up to 70 percent of current levels.  

Industrial symbiosis (IS) is one of the main tools to achieve a circular economy based on synergy between production lines. Instead of a linear production model, with one-time use of raw materials, we may potentially reuse byproducts like plastic, minerals, tires, polyethylene packaging, food and hazardous materials, as inputs in production,  as an alternative to the accumulation of waste and pollution and the  environmental burden they generate.\footnote{For example, the waste of the food industry can be used to produce bio-diesel or compost; agricultural waste can be used as a fertilizer to grow other crops or as an additive to animal feed, to produce bio-gas and even to treat polluted land; the waste from producing animal feed and industrial effluent can both be used to grow algae; copper waste can be used as raw material in metal factories; plastic packaging can serve as raw material in the production of plastic goods; sawdust is used as a substrate for raising animals; correct management of construction sites can reduce waste and increase reuse.} 

Researchers and decision makers describe IS as a win-win strategy that produces environmental value for the public as well as profits for the industry from the commerce in byproducts, and from saving logistic and delivery costs of handling waste (landfill, removal of hazardous materials, incineration or export), saving raw material and avoidance of fines and lawsuits. We argue that these claims cannot be taken at face value, because the implementation of IS may affect the decision of producers how much to produce (and pollute) in the equilibrium. 

Using a simple game-theoretic model, we show that the interest of producers to maximize profits may not coincide with environmental endeavors. This may cause two potential inefficiencies in IS implementation. First, even though at the individual level  there is an economic and environmental justification for adopting IS, high adoption cost and returns to scale may prevent its implementation (\citealp{gerarden2017assessing}). The literature is mostly focused on this inefficiency, which is known as the “gap in energy efficiency”. In line with this literature, our results suggest that the gap in energy inefficiency is more likely to occur when there is a large number of small firms in the market. In such cases, policy may be used to reduce adoption costs, including the costs of acquiring knowledge, transaction, coordination, searching for potential buyers of byproducts, transportation, logistics, regulatory or bureaucratic burden, or modification of the production infrastructure (see \citealp{ashraf2016infrastructure}).

We highlight a second type of inefficiency. When IS is profitable for the producers, the extra profit from trade in byproducts and from savings in environmental taxes induces firms to implement IS and boost their production (and pollution). The rise in production benefits consumers due to price reduction, but at the same time may be environmentally-harmful. Specifically, if the IS technology is highly polluting or highly ineffective in the reuse of byproducts, then it cannot fully offset the increase in pollution generated by the augmented production.\footnote{For example, the recycling of electronics is supposed to protect the environment from the massive amounts of polluting waste created globally each year. Huge quantities of computers, keyboards, cables, and screens are sent to developing countries, where they are recycled. However, during the process of dismantling, the electronic goods are burned off in order to retrieve the copper or aluminum to be sold separately, causing the emission of poisonous gases.} The literature suggests that in some cases households or firms are encouraged to adopt new technologies by means of subsidies, even though the technologies have not proven themselves and have no economic justification (\citealp{gerarden2017assessing}; \citealp{fowlie2018energy}). In these circumstances, it is important to determine whether the benefit outweighs the cost before providing subsidization to these technologies (\citealp{keiser2019consequences}).\footnote{Recycling involves the use of energy in order to modify the waste for its new use and sometimes the recycled material does not replace other production but is added to existing production and encourages greater consumption. Only when the same level of production is maintained, based on greater usage of resources and more efficient use of energy, it is possible to obtain positive environmental outcomes (\citealp{zink2017circular}).}

Importantly, we do not suggest to refrain from IS altogether, but rather address potential inefficiencies and examine how to apply IS successfully. Our recommendation for policy-makers who favor the environment is not to assume that all IS initiatives are win-win strategies. Instead, carefully inspect each IS technology to examine if it is environmentally-harmful in the equilibrium. In case the answer is yes, efforts should focus on research and development to improve the effectiveness in the reuse of byproducts. Our comparative static pinpoints the share of reused byproducts as a key parameter, the only parameter in our framework that holds benefits for firms, consumers, and the environment. In contrast, other policies like targeting the pollution emitted in the production processes or environmental tax policies potentially generate a trade-off between the firms, consumers, and the environment. Thus, we argue that cautious design of the environmental policy may achieve sustainable economic growth in the equilibrium. 

More generally, our contribution lies in the introduction of a thorough framework to describe the countervailing effects of environmental endeavors and analyze their effect on the threshold levels for technology adoption and on the equilibrium outcomes, including the total surplus and the level of pollution. Then, considering the equilibrium outcomes and using comparative static, we map the cases where policy involvement is needed and make policy suggestions. 

The remainder of the paper is organized as follows. Section \ref{LiteratureReview} contains a literature review. Section \ref{TheModel} describes the economic framework, followed by the results of the model in Section \ref{Results}. Section \ref{Discussion} contains concluding remarks and discussion. Proofs are relegated to the Appendix to facilitate the reading. 

\section{Literature Review \label{LiteratureReview}}
An important question in the economic literature considers the economic feasibility of environmental behavior. 
In other words, can a country or a society maintain growth and at the same time protect the environment? Can there be “sustainable environmental growth”? The traditional view holds that environmental regulations have a negative effect on competition, since they place an additional burden on companies. Companies face the direct costs of preventing pollution and the investment required from companies in order to comply with regulations will be at the expense of other profitable opportunities (\citealp{xie2017different}). 

In contrast, an alternative hypothesis, referred to as the Porter Hypothesis, asserts that when environmental regulations are well-designed they encourage innovation that improves the production process, providing a competitive advantage and offsetting the costs of complying with regulations, and therefore in fact contribute to growth (\citealp{porter1995toward}). Innovation facilitates consistent growth while reducing damage to the environment (\citealp{zhang2017can}; \citealp{fernandez2018innovation}).

Many studies have examined the Porter Hypothesis and its validity for various economies. For example, \citet{xie2017different}, \citealp{aghion2016carbon} and \citealp{acemoglu2016transition}) confirm the validity of the Porter Hypothesis in the case of market-based regulations in contrast to command-and-control regulations. \citet{franco2017effect} confirm the validity of the Porter Hypothesis in a number of European countries. 
 \citet{van2017revisiting} find that the Porter Hypothesis is valid in the Dutch economy. While total factor productivity (TFP) increases when regulations encourage ecological innovation oriented toward savings and efficiency in the use of resources, ecological innovations intended to reduce pollution tend to reduce TFP. Our results parallel these findings in the sense that only if IS is highly effective in the reuse of byproducts, its implementation achieves a lower level of pollution as well as higher total surplus in the equilibrium.

\section{The Model \label{TheModel}}
We assume that there are $n$ firms competing in the industry. All the companies are identical and  produce the same product. Each firm $i$ produces a quantity $q_i$ of the product and $Q=\sum_{i=1}^n q_i$ is the total quantity of the product. We make the standard assumption that the marginal cost of production increases with the number of units produced, $cq_i^2$, $c>0$. 
We also assume that the economy is competitive in that each firm alone cannot affect the market price. Let the demand function be given by 
\begin{equation} \label{demand}
    p=a-bQ,
\end{equation}
where $p$ denotes the price of the product and $a,b>0$ are exogenous parameters. 

The production process causes pollution. Pollution denotes waste or byproducts released to the air, water or ground, and we quantify it as a percentage of the final product. Denote the percentage of pollution created per unit of the product by $g$. Accordingly, the production of $Q$ units of the product in the market releases

\begin{equation} \label{poll}
   Poll=g Q
\end{equation}
units of pollution. Firms pay a pollution fine (or environmental tax), by the rate of $d$, such that the total fine paid by each producer is $d g q_i$.\footnote{In many cases, fines are paid on excess pollution. Then, $g$ can be viewed as the percentage of excess pollution above a certain standard, where aiming below the standard is not feasible for the firm. For a discussion on environmental taxes see \citet{acemoglu2016transition},  \citet{barrage2020optimal}.}  
Then, in benchmark case, without IS, the profit of each firm $i$ is given by
\begin{equation} \label{firm_rev}
    \pi_i=q_i(p-dg)-cq_i^2,
\end{equation}
and the consumer surplus, considering the demand function \eqref{demand}, is given by 
\begin{equation} \label{CS}
    CS=\frac{Q(a-p)}{2}, 
\end{equation}
Then, the total surplus is the sum of firms' profits and consumer surplus, 
\begin{equation} \label{CS}
    TS=n \pi_i+CS.
\end{equation}

IS is an environmental endeavor with a unique feature. It provides the firms an opportunity for extra profit if the cost of adoption is not too high. When a firm adopts IS, then it may ceteris paribus reduce pollution by selling a part of the byproducts of its production in the market (to recycle or reuse as raw materials in other production processes). 
 Denote the \textit{effectiveness of industrial symbiosis}, $\alpha$, $0<\alpha<1$, as the percentage of byproducts that can be reused in other production processes (or recycle). Accordingly, if IS is adopted, then a firm sells $\alpha g q_i$ units of byproduct. Consequently, ceteris paribus its level of pollution is reduced to $(1-\alpha) g q_i$.  Note that given the level of production, the pollution (or leftover pollution) is a function of both the percentage of pollution emitted in the production process $g$, or how clean the production technology is, as well as how effective the process of IS is $\alpha$. Thus, a production process that emits many pollutants ($g$ is high), but with proven technologies to use a higher percentage of the pollutants in other production processes ($\alpha$ is high), may be better for the environment than if there are fewer pollutants, but the effectiveness of IS is negligible.
 
 The market price of the byproducts is exogenously given by $p_g$. Accordingly, an IS transaction (the sale of the byproducts in the market) is a source of two gains for a firm. It increases the revenue of the company by $p_g$ and at the same time reduces the level of pollution, which reduces the fine. Nevertheless, implementation of IS involves a fixed cost on the firm, denoted by $c_g$. The cost of adoption of IS includes the cost of acquiring knowledge and searching for potential buyers of the byproducts, modification of the production infrastructure, transportation, logistics, or regulatory or bureaucratic burden.  Accordingly, after adopting IS the profit of firm $i$ is
\begin{equation} \label{symb_rev}
    \pi_i^{symb}=q_i^{symb}[p^{symb}-dg(1-\alpha)+p_g \alpha g]-c(q_i^{symb})^2-c_g.
\end{equation}

Another important parameter to consider is the share of pollution emitted in the process of IS, $k<1$. The process of IS itself may produce pollution of $k \alpha g q_i$,  emitted during transportation of byproducts, their modification to be reusable, or during their reuse in other production processes. Note that $k$ may even be negative if the IS sufficiently reduces the waste of raw materials in other production processes, because of the reuse of byproducts. Accordingly, if IS is adopted, the pollution is the sum of pollution emitted in the production process and the pollution emitted by IS 
\begin{equation} \label{symb_poll}
   Poll_{symb}=(\alpha  k +1-\alpha)g Q^{symb}.
\end{equation}

Recall that our firms are identical. We consider a symmetric equilibrium, where $q_1=q_2=\ldots =q_n=q$, $\pi_1=\pi_2=\ldots=\pi_n=\pi$ without IS, and $q_1^{symb}=q_2^{symb}=\ldots =q_n^{symb}=q^{symb}$, $\pi^{symb}_1=\pi^{symb}_2=\ldots=\pi^{symb}_n=\pi^{symb}$ with IS.

\section{Results \label{Results}}
We start from solving the equilibrium in a benchmark case without IS. Firms maximize their profits \eqref{firm_rev} to obtain the quantity in the equilibrium,

\begin{equation} \label{q_eq} 
    q=\frac{a-dg}{2c+bn}.
\end{equation}
Substituting the quantity produced \eqref{q_eq} in the demand function \eqref{demand} yields the market price 
\begin{equation} \label{p_eq}
    p=\frac{2ac+bndg}{2c+bn}.
\end{equation}
Using the quantity produced \eqref{q_eq} and the market price \eqref{p_eq}, we obtain the firm's profit, consumer surplus, total surplus, and the level of pollution in the equilibrium,
\begin{equation}\label{pi_eq}
    \pi=\frac{c(a-dg)^2}{(2c+bn)^2}.
\end{equation}
\begin{equation}\label{CS_eq}
    CS=\frac{bn^2[a-dg]^2}{2[2c+bn]^2}.
\end{equation}
\begin{equation} \label{ts_symb}
    TS=\frac{n(a-dg)^2}{2[2c+bn]}.
\end{equation}\begin{equation} \label{poll_eq}
    Poll=gq=\frac{g(a-dg)}{2c+bn}.
\end{equation}
We assume hereinafter that $a-dg>0$, such that in the benchmark equilibrium there is a market for the product, and thus firms' quantity \eqref{q_eq}, profit \eqref{pi_eq} and consumer surplus \eqref{CS_eq} are positive. 

It is easy to see from the equilibrium equations \eqref{q_eq}-\eqref{poll_eq} that a larger tax burden $dg$ reduces firms' profits, and thereby drives the firms to reduce production. As a result, the price rises and pushes consumer surplus downwards, and at the same time the pollution level declines. Thus, increasing the fine reduces the level of pollution through a reduction in production, which has a cost for consumers and firms in the equilibrium. 

Our next step is to introduce IS to the market as an alternative way to reduce pollution. Firms decide whether to deviate from the equilibrium and adopt IS based on profit considerations. If an individual firm deviates and adopts IS, given the quantity produced by each firm \eqref{q_eq} and the market price \eqref{p_eq}, its profit \eqref{symb_rev} is given by
\begin{equation} \label{gap_profit}
    \pi^{'}=\pi+q[\alpha g(d+p_g)]-c_g.
\end{equation}
According to equation \eqref{gap_profit}, when a firm decides whether to adopt IS, it weighs the cost of adoption $c_g$ against the profitability gain  $\alpha g(d+p_g)$ per unit of production. Given the reusable share of byproducts $\alpha g$, the company gains $p_g$ from the sale of byproducts in the market as well as a reduction in the fine $d$. Thus, the higher the profitability gain of IS $\alpha g(d+p_g)$ is, it is more likely that firms decide to adopt it. Clearly, if the cost of adoption is negligible $c_g=0$, implementing IS is profitable for the firms, $\pi^{'}> \pi$. Therefore, there is $c_g^{*}$ such that a firm prefers to deviate and adopt IS for $c_g<c_g^{*}$. Formally, a firm is better off by deviating to adopting IS if the cost of adoption $c_g$ is smaller than the profit gain,  
\begin{equation} \label{adoption_cond}
     \pi^{'}>\pi \Leftrightarrow c_g<q[\alpha g(d+p_g)].
\end{equation}
Thus, equation \eqref{adoption_cond} provides a sufficient condition to the adoption of IS. Substituting the quantity produced \eqref{q_eq} into inequality \eqref{adoption_cond}, we obtain that a firm adopts IS when the adoption cost is smaller than a threshold level $c^{*}_g$, where

\begin{defn} \label{cutoff_adoption}
\textit {Cutoff level of adoption cost.}
Define the cutoff level of adoption cost as
\[c^{*}_g=q[\alpha g(d+p_g)]=\frac{(a-dg)\alpha g(d+p_g)}{2c+bn}.\]  
\end{defn} 
Note that the cutoff level of adoption cost is calculated given the quantity produced by the firm in the benchmark case. If we further allow the deviating firm to change its production, then under $c_g<c^{*}_g$ the firm still prefers to deviate for IS. Revealed preference implies that given the price, the new quantity it chooses must increase its profits. Thus this cutoff level is a sufficient, but not a necessary, condition to a firm to deviate from an equilibrium where all firms do not adopt symbiosis.  

Moreover, proposition \ref{equilibrium} states that the cutoff level of adoption cost is a sufficient condition to existence of equilibrium where \textit{all} firms adopt symbiosis.
\begin{proposition} \label{equilibrium} \textit{Adoption of industrial symbiosis}
\newline Assume $c_g<c^{*}_g$. Then there is no symmetric equilibrium where all firms do not adopt IS. Moreover, there is equilibrium where firms adopt IS. 
The cutoff level of adoption cost $c^{*}_g$ increases in $\alpha$ and in $p_g$ and declines in $n$. 
\end{proposition} 
A proof appears in the Appendix.

Therefore, $c_g<c^{*}_g$ is a sufficient condition for the adoption of IS. The last part of proposition \ref{equilibrium} indicates that three parameters positively affect the adoption of IS (by increasing the cutoff level of adoption cost $c^{*}_g$). First, IS becomes more profitable for the firms when the market price of byproducts $p_g$ rises. Second, when the effectiveness of IS $\alpha$ increases, the firm sells a larger share of its byproducts in the market, which provides both more profits and a reduction in the tax burden. Third, the market structure affects the adoption of IS, because IS is typically characterized by returns to scale. A smaller number of firms in the market $n$ implies that each firm has a larger market share, which makes  investing the initial fixed cost of adoption $c_g$ more worthwhile for each firm.

Figure \ref{fig:adoption} illustrates how the number of firms in the market $n$ (the X-axis) affects the profits of each firm (the Y-axis). The blue line presents the firm’s profits in the IS equilibrium and the red line presents the profits in the benchmark case. Surely, a firm’s profits decline in $n$ whether IS is implemented or not. As the number of firms rise, the revenues of each producer decline because of the reduction in its market share. However, the slope is steeper in the case of IS because of the returns to scale, which implies the following. When the number of producers is small, the revenues of each producer are sufficiently large to afford the fixed cost, which makes IS the preferable choice of each producer. Nevertheless, when the number of firms sufficiently rises (firms are sufficiently small), then adopting IS becomes less beneficial for a single firm than the alternative of not adopting IS. 

\begin{figure}
\includegraphics[scale=1]{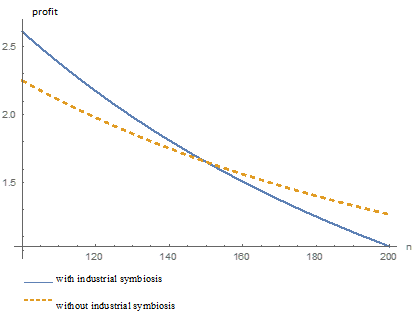}
\caption{Revenues of firms with and without industrial symbiosis. The number of firms in the market $n$ is on the X-axis and the revenues of each firm are on the Y-axis.}
\label{fig:adoption}
\end{figure}

Note that the condition $c_g<c^{*}_g$ stresses a well-known problem in the adoption of environmental solutions, or more generally, the adoption of technologies with external benefits to society, called “gap in energy efficiency”. In some cases, there are efficient technologies that facilitate a saving in energy and would benefit society, but nonetheless are not adopted. The reason is that the company adopting the technology exclusively bears the economic cost of the technology. This feature generates a gap between the benefit to society and the benefit to the firm, and thus the benefit to society is under-priced (\citealp{gerarden2017assessing}). An efficiency gap may call for government involvement to encourage the adoption of IS. 

The case of IS, however, is unique from other environmental endeavors in that firms have internal incentives to implement it and even increase production because of their profitability gain from selling the byproducts and from the reduction in fines. Nevertheless, there is no guarantee that the profitability gain from IS makes the firm fully internalize the benefits to society. In fact, in the case of IS, in addition to the usual problem of under-adoption, the profitability gain may cause a problem of over-adoption. That is, firms may adopt ineffective and highly-polluting technologies of IS that will increase the damage to the environment in the equilibrium. We extend on this point after we solve the equilibrium with IS.  
 
Next, we solve the equilibrium with IS in a similar way to the benchmark case. Maximizing the firm's profit \eqref{symb_rev}, we obtain the quantity in the equilibrium,
\begin{equation} \label{q_symb}
q_{symb}=\frac{a-dg+\alpha g(d+p_g)}{2c+bn}.
\end{equation}
Substituting the quantity produced \eqref{q_symb} in the demand function \eqref{demand} yields the market price 
\begin{equation} \label{p_symb}
p_{symb}=\frac{2ac+bn[dg-\alpha g(d+p_g)]}{2c+bn}.
\end{equation} 
We relegate the solution of the rest of the equilibrium ($CS_{symb}$, $\pi_{symb}$, and $TS_{symb}$) to the Appendix. We further verify in the Appendix that a single firm does not wish to deviate from the IS equilibrium after it is implemented, namely given that the adoption cost is lower than the threshold $c^{*}_g$ (recall definition \eqref{cutoff_adoption}).

 Next, we analyze the implications of IS. To facilitate notation, denote the gaps between the benchmark case and the IS equilibrium in production, price, consumer surplus, firm's profit, and total surplus by $q_{symb}-q=\Delta q$, $p-p_{symb}=\Delta p$, $CS_{symb}-CS=\Delta CS$, $\pi_{symb}-\pi=\Delta \pi-c_g$, $TS_{symb}-TS=\Delta TS$, respectively. In proposition \ref{comparison}, we summarize the comparison between the two equilibrium, 

\begin{proposition} \label{comparison} \textit{Comparison of equilibrium with and without industrial symbiosis}
\newline (1) $q_{symb}>q$, $p_{symb}<p$, $CS_{symb}>CS$. 
\newline (2) Assume $c_g<c^{*}_g$ and $\alpha g(d+p_g)c>bn(a-dg)$. Then, $\Delta \pi>0$, $\Delta TS>0$. 
\newline (3) $\Delta q$, $\Delta p$, $\Delta CS$, $\Delta \pi$, $\Delta TS$ increase in $\alpha$, $d$, and $p_g$. $\Delta q$, $\Delta p$ and $\Delta CS$ increase in $g$.
\end{proposition}

According to proposition \ref{comparison}, the adoption of IS leads to an increase in the quantity produced. When firms can sell a share of $\alpha$ of their byproducts for reuse, it becomes more profitable to increase production both because of the price they receive per unit $p_g$ and because of the reduction in fines. In other words, the higher the profitability gain of IS $\alpha g(d+p_g)$ is, the larger is the incentive of firms to increase production. The increased supply leads to a reduction in the price of the product, which is beneficial for consumers. Therefore, the profitability gain of IS $\alpha g(d+p_g)$ plays a crucial role for the total surplus of firms and consumers. 

To complete the analysis, we turn to measure how IS affects the level of pollution. Recall that the primer goal of IS is to reduce pollution. To examine whether this goal is achieved in the equilibrium, we compare the pollution with and without IS (equations \eqref{poll} and \eqref{symb_poll}). To facilitate notation, let us define the gap in pollution caused by IS as $\Delta Poll=Poll_{symb}-Poll$,  
\begin{equation} \label{comp_poll}
     \Delta Poll= gn[(q_{symb}-q)-q_{symb} \alpha (1-k)]
\end{equation}
which implies,  
\begin{equation} \label{comp1_poll}
     \Delta Poll>0\Leftrightarrow (q_{symb}-q)>q_{symb} \alpha (1-k).
\end{equation}

Inequality \eqref{comp1_poll} indicates that two counter effects play an essential role in determining how IS affects pollution.
First, focusing on the right-hand side of inequality \eqref{comp1_poll}, IS is entitled to reduce pollution depending on the quality of the technology $\alpha (1-k)$ (that in turn depends on its effectiveness $\alpha$ and on how polluting the process is $k$). A sufficiently effective IS (a large percentage of reusable byproducts $\alpha$) and a more green one (emits a lower level of pollution $k$), successfully reduces the level of pollution relative to the benchmark case. However, there is a counter effect of IS on pollution, emanating from the decisions of firms in the equilibrium. This effect is captured by the left-hand side of inequality \eqref{comp1_poll}. Paradoxically, the profitability gain of IS (from selling byproducts $p_g$ and from the lower fines) induces an increase in production ($q_{symb}>q$ by proposition \ref{comparison}). The increased production is liable to bring about more pollution and more harm to the environment. This somewhat surprising negative effect of IS on pollution overcomes the first effect for certain parameters of the model, i.e., in certain industries and  for certain technologies, 

\begin{proposition} \label{pollution} \textit{Comparison of pollution with and without industrial symbiosis}
\newline (1) $\Delta Poll>0$, if $1-k\rightarrow 0$.
\newline (2) $\Delta Poll>0 \Leftrightarrow \frac{1}{1-k}-\alpha>\frac{a-dg}{g(d+p_g)}$. Else, $\Delta Poll \leq 0$. 
\newline (3) $\Delta Poll$ increases in $g,d,p_g$, and $k$, and declines in $\alpha$.
\end{proposition}

The first part of proposition \ref{pollution} focuses on the quality of the process of IS. When IS is of a very low quality because the process itself is very polluting, $k\rightarrow 1$, then the technology of IS almost cannot reduce pollution in the equilibrium. In this case, a technology barrier in the process of IS brings more pollution when companies adopt IS than in the benchmark. The second part of proposition \ref{pollution} presents a necessary and sufficient condition for an increased pollution when IS is implemented. If the quality level of IS $1-k$ is sufficiently low relative to the increased production caused by the augmented profitability that IS brings, then IS boosts the level of pollution in the equilibrium. Otherwise, IS reduces pollution and improves the quality of the environment. The last part of proposition \ref{pollution} disassembles the two effects of IS on pollution to their components. On the left-hand side we have the technology parameters, $\alpha$ and $k$, and right-hand side contains parameters of firms' profitability gain from selling their byproducts in the market and from the reduction in the tax burden. Specifically, the profitability gain of firms increases when production is more polluting $g$, or when the price of byproducts $p_g$ or the tax rate $d$ are larger.

Important policy implications emanate from our propositions and specifically from the comparative static. Adopting IS may create a trade off between the total surplus and pollution. For example, focusing on byproducts with a high price in the market (large $p_g$), or increasing the fine on pollution $d$, increases the profits of firms relative to the benchmark case $\Delta \pi$, and in turn increases consumer surplus $\Delta CS$ and the total surplus $\Delta TS$ when the industry switches to IS. However, at the same time these parameters increase the level of pollution when IS is implemented $\Delta Poll$, because of the augmented production $\Delta q$. Therefore, at the end the environment may be worse off by IS due to high levels of parameters that compose the profitability gain, particularly high $g(d+p_g)$. 

Surely, the social welfare should consider both the total surplus and the level of pollution. Therefore, our results suggest that policy-makers who wish to direct the market into a larger surplus and a lower level of pollution should focus on supporting advances in the quality of IS $\alpha (1-k)$, namely the effectiveness and cleanness of the process. That is, policy-makers should insist on and encourage the implementation of IS technologies that are more effective and green and invest in the innovation of such technologies. In contrast to the other parameters, a larger quality of IS reduces the level of pollution without generating a trade off between the total surplus and pollution. Moreover, the effectiveness of IS (the percentage of byproducts that can be reused, $\alpha$) is a key parameter that benefits all parties. A larger $\alpha$ generates a larger surplus for firms and consumers as well as a lower level of pollution. Therefore, our results suggest that policy effort should concentrate on increasing the quality of IS in order to reduce the level of pollution without harming the total surplus. 

To illustrate our results, and specifically the countervailing effects of IS and the key importance of the effectiveness of IS, we present three simulations of the model\footnote{A preliminary sketch of the model and simulations appear in Hebrew in Heth Foundation series Mehkarei Regulazia}. In each simulation, we calculate the level of pollution (the Y-axis) for different levels of prices of the byproducts in the market $p_g$ (the X-axis) with IS (the blue line) and without IS (the red line). The only difference between the simulations is the level of $\alpha$. In figure \ref{badtech}, the technology of IS is of low effectiveness, only 1\% of the byproducts can be reused $\alpha=0.01$; in figure \ref{goodtech}, we take high effectiveness of IS, 90\% of the byproducts are reusable, $\alpha=0.9$; and in figure \ref{medtech} we assume medium effectiveness of the process, 50\% reusable byproducts, $\alpha=0.5$.

\begin{figure}
\includegraphics[scale=1]{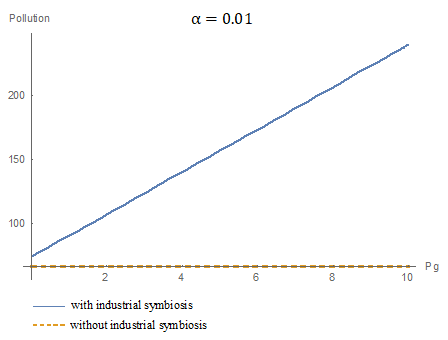}
\caption{Pollution with and without IS for $\alpha=0.01$. The market price of byproducts in the market $p_g$ is on the X-axis and the pollution is on the Y-axis.}
\label{badtech}
\end{figure}

\begin{figure}
\includegraphics[scale=1]{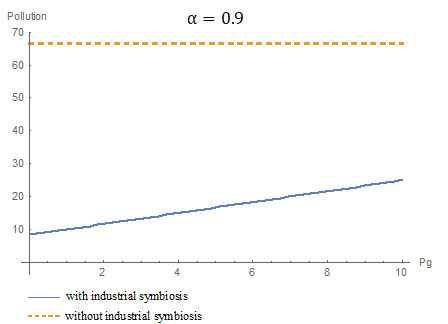}
\caption{Pollution with and without IS for $\alpha=0.9$. The market price of byproducts in the market $p_g$ is on the X-axis and the pollution is on the Y-axis.}
\label{goodtech}
\end{figure}

\begin{figure}
\includegraphics[scale=1]{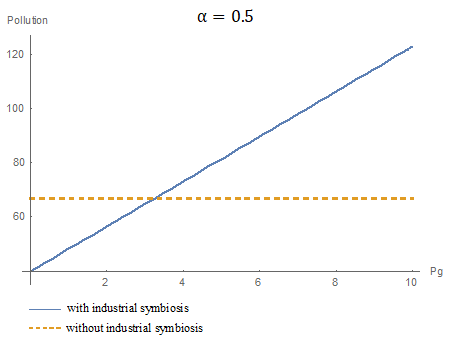}
\caption{Pollution with and without IS for $\alpha=0.5$. The market price of byproducts in the market $p_g$ is on the X-axis and the pollution is on the Y-axis.}
\label{medtech}
\end{figure}

In figure \ref{badtech} the IS is highly ineffective, such that byproducts can hardly be reused. Therefore, the  implementation of IS increases the level of pollution because of the increased production, for any price of the byproducts. 
Figure \ref{goodtech} presents the opposite case, high effectiveness of IS, where 90\% of byproducts can be reused. With highly effective process, the implementation of IS compensates for the increased production and reduces the level of pollution, for any price of the byproducts. 

Comparing figures \ref{badtech}-\ref{goodtech}, in both figures when IS is implemented a rise in the market price of byproducts always augments profitability leading to increased production, and thereby to more pollution. However, in figure \ref{goodtech} the slope is less steep. The reason is that a highly effective IS reuses most of the pollution (90\%), reducing the leftover pollution relative to the case in figure \ref{badtech} of ineffective process (with only 1\% of reusable byproducts).  

Figure \ref{medtech} illustrates the offsetting effects of IS on pollution in an industry with medium effectiveness of IS, 50\% reuse of byproducts. In this case, the effect of IS on the level of pollution depends on the price of byproducts. If the price is low, then implementation of IS reduces pollution, whereas a high price reverses the result. The reason is that a high price of byproducts is highly profitable for the firms, causing a massive boost in production as well as pollution, and thus harms the environment relative to the benchmark case. 

It is important to note that the technology parameters of IS are expected to enhance over time with the development of new technologies friendlier to the environment. In practice, standards of emissions for older factories are less rough and allow the use of the 'Best Practical Technology' (BPT) so the industry remains profitable. Over time, as the technological barriers decline, the standards are upgraded and factories are enforced to use the ‘Best Available Technology', BAT, the average between the more polluting technologies and   greener advanced technologies in newer factories. Accordingly, we can interpret figures \ref{badtech}-\ref{medtech} as describing the effect of technological changes leading to improvements in the effectiveness of IS. Figure \ref{badtech} describes an initial low effective IS with augmented pollution in the equilibrium. In time, investment in R\&D promotes the innovation of more effective technologies, where IS may increase or decrease pollution depending on the price of byproducts in the market, as Figure \ref{medtech} describes. Eventually, technologies in the industry are sufficiently environment-friendly, such that their adoption reduces pollution, as Figure \ref{goodtech} illustrates. 

Thus, a lesson from our results is that IS should be adopted with caution,  considering all available IS technologies and their effect on the equilibrium. When the IS technology is sufficiently green and effective, the regulator's role should focus on reducing the cost of adoption to assure its implementation. However, in certain industries with low quality of available technologies, policy should provide incentives for technological improvements, so that in the equilibrium, considering the reaction of all parties involved, IS will satisfy its entitled goal - to enhance sustainability.  

\section{Conclusion and discussion \label{Discussion}}
Our study shows that industrial symbiosis (IS) is not a "magic word" for sustainability. The fact that IS is adopted does not guarantee the reduction of environmental harm.  In choosing a legal platform for IS, it is important to examine all the effects in the equilibrium, including how green and effective the suggested process is and its effect on the level of production in the equilibrium. A thorough analysis of the technologies in the market is crucial to nail down industries where IS can improve sustainability, as opposed to industries where IS may harm the environment, thereby further technological progress is necessary before implementation.   

Our findings show that IS is beneficial for the environment in case the quality of the process is sufficiently high (the technology of IS is sufficiently green and effective, such that a high proportion of byproducts is converted for circular use). Then, although firms boost production (and pollution) because of their profitability gain (from selling their byproducts and from saving in environmental tax), a high quality IS can compensate for the extra pollution and reduce the level of pollution in the equilibrium. In this case, policy should focus on reducing the cost of adoption in order to guarantee the  implementation of IS.    

In contrast, a low quality IS technology may harm the environment because it cannot offset the negative effect of the increased production on the environment. In this case, our policy recommendation is not to encourage  adoption of the existing IS technologies, but rather to incentivize the research and development of green and effective technologies that will make IS  environmentally-friendly in the equilibrium. 

Our results suggest that encouraging the development of technologies that reuse an augmented share of byproducts increases the total surplus and enhances sustainable economic growth, in contrast to other suggested policies. Incentives for R\&D may include government subsidies, environmental taxes (market-based regulation) or other methods, like enforcing companies to use part of their profits for R\&D or imposing higher environmental standards on firms (see discussion on command-and-control mechanisms in \citealp{duflo2018value}). While in the absence of IS the regulator may compromise on Best Practicable Technology (BPT) to alleviate the economic burden on a company, an increase in the firms' profitability following the adoption of IS provides an opportunity to require the implementation of he Best Available Technology (BAT), or the development of even greener and more effective technologies. For a related literature on the advantages of market-based regulations over command-and-control regulations see the end of Section \ref{LiteratureReview}. Another way to push firms to develop these technologies is informal regulation, such as providing information to the public by announcing the names of polluting companies in the media or consumer boycotts. It appears that there is willingness among the public to be part of the effort to protect the environment. In Switzerland, about 70 percent of consumers would prefer goods to be produced with less carbon emissions and a large proportion would even agree to pay more for them (\citealp{blasch2014context}). 

\bibliography{circ}
\bibliographystyle{chicago}

\section*{Appendix}
\textbf{The equilibrium when industrial symbiosis is adopted.} 
Using the quantity produced \eqref{q_symb} and the market price \eqref{p_symb}, we obtain the firm's profit, consumer surplus, total surplus, and the level of pollution in the equilibrium, respectively,
\begin{equation}\label{pi_symb_eq}
    \pi_{symb}=\frac{c[a-dg+\alpha g(d+p_g)]^2}{[2c+bn]^2}-c_g
\end{equation}
\begin{equation*} \label{cs_symb}
    CS_{symb}=\frac{bn^2[a-dg+\alpha g(d+p_g)]^2}{2[2c+bn]^2}
\end{equation*}
\begin{equation} \label{ts_symb}
    TS_{symb}=\frac{n[a-dg+\alpha g(d+p_g)]^2}{2[2c+bn]}-nc_g
\end{equation}
\begin{equation} \label{poll_symp_eq}
     Poll_{symb}= \frac{(\alpha  k +1-\alpha)ng[a-dg+\alpha g(d+p_g)]}{2c+bn}.
\end{equation}

Next, we verify that a single firm does not wish to deviate from the IS equilibrium. 
Given the quality produced \eqref{q_symb} and the market price \eqref{p_symb} in the IS equilibrium, the profit of a firm that deviates to not adopting IS is given by
\begin{equation} \label{pi_dev_symb}
    \pi_{symb}^{'}=\pi_{symb}-q_{symb}[\alpha g(d+p_g)]+c_g.
\end{equation}
According to equation \eqref{pi_dev_symb}, for $c_g=0$, $\pi_{symb}> \pi_{symb}^{'}$, therefore there is a threshold level  $c_g^{'}$, such that a firm prefers not to deviate for $c_g<c_g^{'}$.
Accordingly, a firm is better off by not deviating if 
\begin{equation} \label{adoption_cond_symb}
     \pi_{symb}> \pi_{symb}^{'} \Leftrightarrow c_g<c_g^{'}.
\end{equation}
where $c_g^{'}=q_{symb}[\alpha g(d+p_g)]$. Next, we prove that $c_g<c_g^{'}$ is not binding, given our assumption in proposition \ref{equilibrium} that $c_g<c^{*}_g$, where $c^{*}_g=q[\alpha g(d+p_g)$ (recall definition \eqref{cutoff_adoption}). From the equilibrium  equations \eqref{q_eq} and \eqref{q_symb}, we deduce that $q_{symb}>q$, because $q_{symb}=q+\Delta q$ and $\Delta q=\frac{\alpha g(d+p_g)}{2c+bn}>0$. As a result, $c^{*}_g<c_g^{'}$. This ends the proof with no further assumption needed, because $c_g<c^{*}_g$ and $c^{*}_g<c_g^{'}$ imply that $c_g<c_g^{'}$, and no firm deviates. 

\begin{proof}[Proof of proposition \ref{equilibrium}]
A firm deviates and adopts IS when $\pi^{'}>\pi$. Substituting the quantity produced \eqref{q_eq} in the inequality \eqref{adoption_cond}, we obtain that firms adopt IS when $c_g<c^{*}_g$, where $c^{*}_g=(\frac{a-dg}{2c+bn})[\alpha g(d+p_g)]$. As we show above, $c_g<c^{*}_g$ is sufficient to existence of equilibrium where all firms adopt IS.  Additionally, it is easy to see that $c^{*}_g$ monotonically increases in $\alpha$ and in $p_g$ and declines in $n$. 
\end{proof}

\begin{proof}[Proof of proposition \ref{comparison}]
\begin{enumerate}
    \item Using the equilibrium equations with and without IS, we obtain that $q_{symb}>q$ because $\Delta q=\frac{\alpha g(d+p_g)}{2c+bn}>0$,
$p_{symb}<p$, because $\Delta p=\frac{bn[\alpha g(d+p_g)]}{2c+bn}>0$, and
$CS_{symb}>CS$, because $\Delta CS=\frac{b[n\alpha g(d+p_g)]^2}{2[2c+bn]^2}>0$.
\item $\pi_{symb}-\pi=\Delta \pi-c_g$, where $\Delta \pi=\frac{c[(\alpha g(d+p_g))^2+2\alpha g(d+p_g)(a-dg)]}{[2c+bn]^2}$. 
Therefore, $\pi_{symb}>\pi \Leftrightarrow c_g<\Delta \pi$.

A sufficient condition for $c_g<\Delta \pi$, given that $c_g<c^{*}_g$, is $c^{*}_g<\Delta \pi$. Substituting $c^{*}_g$ (see definition 1) and $\Delta \pi$ and rearranging yields a sufficient condition of $\alpha g(d+p_g)c>bn(a-dg)$. Then, from $\Delta CS>0$ and $\Delta \pi>0$ we conclude that $\Delta TS>0$. 
\item $\frac{\partial \Delta \pi}{\partial d}=\frac{2c\alpha g [(a-dg)+(d+p_g)(1-g)]}{[2c+bn]^2}$. This expression is positive, given that the share $g<1$ and our assumption that $a-dg>0$. The rest of the proof is straightforward. 
\end{enumerate}
\end{proof}

\begin{proof}[Proof of proposition \ref{pollution}]
\begin{enumerate}
    \item According to inequality \eqref{comp1_poll}, $\Delta Poll>0 \Leftrightarrow (q_{symb}-q)>q_{symb} \alpha (1-k)$. If $1-k\rightarrow 0$, then the right-hand side tends to 0, and it is left to prove that the left-hand side is positive. This is true because $q_{symb}>q$ (recall proposition \ref{comparison}). 
    \item Substituting $q_{symb}$ and $q$ (equations \eqref{q_eq}, \eqref{q_symb}) in inequality \eqref{comp1_poll}, we obtain that $\Delta Poll>0\Leftrightarrow \frac{\alpha gn}{2c+bn}[g(d+p_g)-(1-k)[a-dg+\alpha g (d+p_g)]]>0$. Rearranging this inequality yields that $\frac{1}{1-k}-\alpha>\frac{a-dg}{g(d+p_g)}$. Deriving this inequality by $g,d,p_g,\alpha$, and $k$, part (3) is straightforward.
\end{enumerate}
\end{proof} 

\end{document}